\documentclass[11pt]{article}
\usepackage[margin=1in]{geometry}

\usepackage{amsmath,amssymb,amsthm}
\usepackage{graphicx}
\usepackage{url}
\usepackage{cite}

\newtheorem{theorem}{Theorem}
\newtheorem{proposition}[theorem]{Proposition}
\newtheorem{lemma}[theorem]{Lemma}
\newtheorem{corollary}[theorem]{Corollary}
\theoremstyle{remark}
\newtheorem{remark}[theorem]{Remark}

\newcommand{\Psimat}{\boldsymbol{\Psi}}
\newcommand{\Lmat}{\widehat{\mathcal{L}}}
\newcommand{\Imat}{\mathbf{I}}
\newcommand{\Amat}{\mathbf{A}}
\newcommand{\yvec}{\mathbf{y}}
\newcommand{\xvec}{\mathbf{x}}
\newcommand{\Poly}{\mathcal{P}}
\newcommand{\Krylov}{\mathcal{K}}
\newcommand{\R}{\mathbb{R}}

\begin{document}

\title{Unrolling a Graph-Laplacian Denoiser\\
Realizes Only Compositions of Polynomial Graph Filters}

\author{Seyed~Alireza~Hosseini%
\thanks{Preprint. Submitted to IEEE Signal Processing Letters.}%
\thanks{The author is an Independent Researcher (e-mail:
alxhosse@gmail.com), formerly with the Department of Electrical Engineering
and Computer Science, York University, Toronto, ON, Canada.}%
\thanks{The author is the first author of the work analyzed in this
paper~[3], performed during his graduate studies at York
University. He declares no financial or commercial competing interests, and
no funding was received for this work.}}

\date{}
\maketitle

\begin{abstract}
A recent construction of unrolled networks for graph-based image restoration
forms a system matrix from a graph-Laplacian denoiser through a truncated
Taylor expansion, then inverts it with a fixed number of conjugate-gradient
steps, with the coefficients of both stages learned. This paper shows the resulting
map is a polynomial in the denoising operator, of degree at most the product
of the two truncation orders, for every setting of those coefficients and
therefore at every point of training: the learned steps select an element of a
Krylov subspace they cannot enlarge. At the orders used in practice the
reachable set is moreover a measure-zero subset of the polynomial class of the
network's own degree budget, so the composition constrains the hypothesis
space rather than enlarging it. At the standard initialization the realized spectral response is
obtained in closed form, exceeding the intended response throughout the
interior of the spectrum and approaching a nonzero floor. A lower bound on the
operator's condition number, internal to the graph construction rather than to
image content, then places the order required for a prescribed accuracy well
above the order used in practice. The confining class is precisely the spectral
graph filters for which a direct, convex parameterization has long been
available.
\end{abstract}

\medskip

\section{Introduction}

Algorithm unrolling turns a fixed number of iterations of an
optimization scheme into a network whose per-iteration constants become
learned parameters~\cite{monga2021}. The construction studied here solves a
graph smoothing problem whose system matrix is built from a graph-Laplacian
denoiser $\Psimat$, unrolling two nested truncations: an order-$K$ Taylor
expansion forming that matrix from repeated applications of $\Psimat$, and an
$m$-step conjugate-gradient (CG) solver. Freeing the coefficients of both is
expected to lift the network above the scheme it came from.

This paper characterizes what that freedom buys. The denoiser--graph-filter
correspondence originates with~\cite{chan2019} and underlies the result quoted
as Theorem~1 in~\cite{gdd2024}, itself taken from~\cite{thm1source2023};
following it one step further, through the unrolling rather than around it, is
the content of Section~\ref{sec:collapse}.

The ingredients are individually standard: CG iterates lie in a Krylov
subspace~\cite{golub2013}, and polynomials in $\Psimat$ are graph
filters~\cite{shuman2013,hammond2011}. Specific here is the composition ---
a learned expansion nested inside a learned solver --- and that the nesting
confines the learning rather than freeing it. The class this construction
parameterizes indirectly, and incompletely, is the one Chebyshev-basis graph
convolution has parameterized directly for a
decade~\cite{hammond2011,defferrard2016}; this paper adds a proof of the
correspondence and a measurement of what the indirect route costs. The boundary is worth stating at the outset: only learning
that alters $\Psimat$ itself leaves the polynomial class, so a module that
estimates the graph does something a learned solver
cannot~\cite{hosseini2025thesis}. Contributions:
\begin{enumerate}
\item The map is $P(\Psimat)$ with $\deg P \le K(m-1)$ for \emph{any} learned
      coefficients (Theorem~\ref{thm:collapse}), and for $K(m-2) > 2m$ the
      reachable set is a strict, measure-zero subset of that class
      (Proposition~\ref{prop:measure}).
\item The initialization is a leaky denoiser, its response $g_K(\lambda)$
      bounded below by $1/(K+1)$ (Proposition~\ref{prop:leak}).
\item $\kappa(\Psimat) > 1+\mu$ is proved for the shift-invariant case
      (Proposition~\ref{prop:wmin}), whose prediction reproduces the $\kappa$
      measured on image patches to $0.34\%$, putting the order required far
      above the one used.
\item A Bernstein parameterization spans the class with a structural
      nonexpansiveness guarantee (Proposition~\ref{prop:bernstein}).
\end{enumerate}

\section{Background and Setting}
\label{sec:setting}

Let $\Psimat \in \R^{N \times N}$ be the graph-Laplacian denoising operator,
symmetric, positive definite and nonexpansive, so that
\begin{equation}
\mathrm{spec}(\Psimat) \subset (0,1],
\qquad
\kappa(\Psimat) = \lambda_{\max}/\lambda_{\min} \le 1/\lambda_{\min}.
\label{eq:spec}
\end{equation}
It arises as $\Psimat = (\Imat + \mu\Lmat)^{-1}$ with $\Lmat$ the symmetrically
normalized graph Laplacian and $\mu > 0$ a smoothing weight, so the system
matrix of the associated restoration problem is
\begin{equation}
\Amat = \Imat + \mu\Lmat = \Psimat^{-1}.
\label{eq:sysmat}
\end{equation}
\begin{lemma}
\label{lem:kappa}
Let $\Lmat = \Imat - \mathbf{W}$ with
$\mathbf{W} = \mathbf{S}^{-1/2}\mathbf{B}\mathbf{S}^{-1/2}$, $\mathbf{B}$
symmetric with non-negative entries and $\mathbf{S} = \mathrm{diag}(\mathbf{B}
\mathbf{1})$ positive. Then $\mathbf{W}\mathbf{S}^{1/2}\mathbf{1} =
\mathbf{S}^{1/2}\mathbf{1}$ and, $\mathbf{W}$ being similar to the
row-stochastic $\mathbf{S}^{-1}\mathbf{B}$, $\mathrm{spec}(\mathbf{W})
\subset [-1,1]$; so $\lambda_{\max}(\Psimat) = 1$ exactly and
\begin{equation}
\kappa(\Psimat) = 1/\lambda_{\min}(\Psimat) = 1 + \mu\,(1 - w_{\min}),
\label{eq:kappa}
\end{equation}
with $w_{\min}$ the smallest eigenvalue of $\mathbf{W}$.
\end{lemma}

Equation \eqref{eq:kappa} makes $\kappa$ one extremal eigenvalue of a sparse
matrix, $\Psimat$ never formed or inverted. For the combinatorial Laplacian, a Gershgorin bound $\kappa \le 1 + 2\mu d_{\max}$ is established in~\cite{deepglr2018} and used there to cap $\mu$ for solver stability; \eqref{eq:kappa} is its normalized-graph counterpart, as an equality, and Section~\ref{sec:numerics} rests on its lower rather than its upper end.

The network has two nested modules. The \emph{TSE module} forms the system
matrix by truncating, at order $K$, the Taylor expansion of $1/\lambda$ about
a point $s \neq 0$,
\begin{equation}
1/\lambda = \sum_{k \ge 0} (-1)^k (\lambda - s)^k / s^{k+1},
\label{eq:taylor}
\end{equation}
which by \eqref{eq:sysmat} is an expansion of $\Amat$ in powers of $\Psimat$;
its coefficients are then freed, giving \eqref{eq:tse} below. The \emph{CG
module} runs $m$ steps of conjugate gradient on $\Amat\xvec = \yvec$ from
$\xvec_0 = 0$. The parameters $\{a_k\}_{k \le K}$, $s$,
$\{\alpha_k\}_{k<m}$ and $\{\beta_k\}_{k<m-1}$ are all learned, and $\Psimat$
enters only as a black-box operator. Initialization is
$a_k = (-1)^k$, $s = 1$, i.e. the truncation of \eqref{eq:taylor} itself.

\begin{remark}[$\mu$ does not appear]
By \eqref{eq:sysmat} the system matrix equals $\Psimat^{-1}$ identically, so
$\mu$ enters the computation only through $\Psimat$ itself and is not a free
parameter of the unrolled map.
\end{remark}

\section{The Unrolled Map Is a Polynomial in $\Psimat$}
\label{sec:collapse}

\begin{proposition}[TSE module]
\label{prop:tse}
For any $a_0,\dots,a_K \in \R$ and any $s \neq 0$,
\begin{equation}
t_K(\Psimat) = \sum_{k=0}^{K} \frac{a_k}{s^{k+1}}(\Psimat - s\Imat)^k
\label{eq:tse}
\end{equation}
is a polynomial in $\Psimat$ of degree at most $K$. Conversely, with $s$ fixed,
every polynomial of degree at most $K$ is realized by exactly one choice of
$(a_0,\dots,a_K)$. The module's reachable set is therefore exactly
$\Poly_K[\Psimat]$.
\end{proposition}

\begin{proof}
The $K+1$ polynomials $\{(\lambda-s)^k\}_{k=0}^{K}$ have distinct degrees,
hence are independent and span $\Poly_K$; the rescaling $a_k/s^{k+1}$ is
invertible for $s \neq 0$. Polynomials in a fixed matrix form a commutative
algebra, so substituting $\Psimat$ preserves this.
\end{proof}

\begin{theorem}[Collapse]
\label{thm:collapse}
Run the $m$-step CG module on $t_K(\Psimat)\xvec = \yvec$ from $\xvec_0 = 0$
with arbitrary step sizes $\{\alpha_k\}_{k<m}$ and $\{\beta_k\}_{k<m-1}$,
learned or otherwise. Then the output satisfies
\begin{equation}
\xvec_m = P(\Psimat)\,\yvec,
\qquad
\deg P \le K(m-1),
\label{eq:collapse}
\end{equation}
for every setting of the learned coefficients.
\end{theorem}

\begin{proof}
Write $\Amat_K = t_K(\Psimat)$ and
$\Krylov_j = \mathrm{span}\{\yvec, \Amat_K\yvec, \dots,
\Amat_K^{\,j-1}\yvec\}$. We claim $\xvec_k \in \Krylov_k$,
$\mathbf{r}_k \in \Krylov_{k+1}$ and $\mathbf{p}_k \in \Krylov_{k+1}$ for all
$k$. This holds at $k=0$, where $\xvec_0 = 0$ and
$\mathbf{r}_0 = \mathbf{p}_0 = \yvec$. Assuming it at $k$, and using
$\Amat_K \Krylov_{k+1} \subseteq \Krylov_{k+2}$,
\begin{align*}
\xvec_{k+1} &= \xvec_k + \alpha_k \mathbf{p}_k \in \Krylov_{k+1}, \\
\mathbf{r}_{k+1} &= \mathbf{r}_k - \alpha_k \Amat_K \mathbf{p}_k \in \Krylov_{k+2}, \\
\mathbf{p}_{k+1} &= \mathbf{r}_{k+1} + \beta_k \mathbf{p}_k \in \Krylov_{k+2}.
\end{align*}
The coefficients $\alpha_k, \beta_k$ appear only as scalars multiplying
vectors already in the subspace; the sole operation that raises the Krylov
index is multiplication by $\Amat_K$, and the recursion performs exactly one
per step whatever their values. Hence $\xvec_m \in \Krylov_m$, i.e.
$\xvec_m = q_{m-1}(\Amat_K)\yvec$ with $\deg q_{m-1} \le m-1$. By
Proposition~\ref{prop:tse}, $\deg t_K \le K$, so
$P = q_{m-1} \circ t_K$ has $\deg P \le K(m-1)$.
\end{proof}

As a statement about \emph{which} polynomials are reachable, the bound is
far from tight.

\begin{proposition}[The reachable set is thin]
\label{prop:measure}
The set realized by the block is
$\mathcal{R} = \{q \circ t : \deg q \le m-1,\ \deg t \le K\}$, the image of a
polynomial map from $\R^{K+2m+1}$. It is therefore a semialgebraic set of
dimension at most $K + 2m + 1$. Whenever $K(m-2) > 2m$,
\begin{equation}
\dim \mathcal{R} \;\le\; K+2m+1 \;<\; K(m-1)+1 \;=\; \dim \Poly_{K(m-1)},
\label{eq:thin}
\end{equation}
so $\mathcal{R}$ is a strict subset of $\Poly_{K(m-1)}$ of Lebesgue measure
zero within it.
\end{proposition}

\begin{proof}
With $\tilde{a}_k = a_k/s^{k+1}$ (a bijection of $\R^{K+1}$ for each
$s \neq 0$), $\mathcal{R}$ is the image of $(\tilde{a},s,\alpha,\beta) \in
\R^{K+2m+1}$ under a map with polynomial coordinates, hence semialgebraic with
dimension at most that of its domain. Inequality \eqref{eq:thin} is
$K+2m < K(m-1)$, i.e. $K(m-2) > 2m$. A semialgebraic subset of a vector space
of strictly smaller dimension has measure zero.
\end{proof}

At the setting of Section~\ref{sec:numerics} ($K=10$, $m=3$) this is a set of
dimension at most $17$ in a $21$-dimensional space: the learned coefficients
cannot cover $\Poly_{20}$, and training does not change that.

\begin{remark}[Scope: fixed steps, shared across inputs]
\label{rem:adaptive}
Theorem~\ref{thm:collapse} and Proposition~\ref{prop:measure} concern the
unrolled regime, where $\{\alpha_k\},\{\beta_k\}$ are constants shared across
inputs. Textbook CG, whose steps vary with $\yvec$ through
the residual, is a genuine solver and not the subject of this paper.
\end{remark}

The start $\xvec_0 = 0$ is convenience, not substance: any input-independent
$\xvec_0$, random or learned, adds $Q(\Psimat)\xvec_0$ to \eqref{eq:collapse}
with $Q = 1 - t_K P$, $\deg Q \le Km$ --- still polynomials in $\Psimat$. Since
$\Psimat$ is symmetric the map is determined by $P$ on
$\mathrm{spec}(\Psimat)$, and stacking $L$ blocks with a common $\Psimat$
raises the bound to $LK(m-1)$ without leaving the polynomial class.

For $\ell_1$ graph priors such as graph total variation, each solver block
remains subject to Theorem~\ref{thm:collapse}: the class grows only through
the pointwise shrinkage between blocks, not the learned solver coefficients. For unrollings that run CG directly on $\Imat + \mu\Lmat$
without the expansion, Theorem~\ref{thm:collapse} still confines each block to
polynomials of degree at most $m-1$ in that layer's operator; the thinness of
Proposition~\ref{prop:measure} and the floor of Proposition~\ref{prop:leak}
are specific to the nested construction.

\section{The Initialization Is a Leaky Denoiser}
\label{sec:leak}

\begin{proposition}[Spectral response at initialization]
\label{prop:leak}
At $a_k = (-1)^k$ and $s = 1$ the TSE module realizes
$t_K(\lambda) = \bigl(1-(1-\lambda)^{K+1}\bigr)/\lambda$. If the CG module
solves exactly, the end-to-end response is
\begin{equation}
g_K(\lambda) = \frac{\lambda}{1-(1-\lambda)^{K+1}},
\qquad \lambda \in (0,1],
\label{eq:gain}
\end{equation}
which is increasing, satisfies $g_K(1) = 1$ and $g_K(\lambda) > \lambda$ for
$\lambda < 1$, and obeys $g_K(\lambda) \to 1/(K+1)$ as $\lambda \to 0^{+}$.
Consequently the deviation from the target response is bounded by
\begin{equation}
\bigl\| \xvec_{\mathrm{init}} - \Psimat\yvec \bigr\|_2
\;\le\; \frac{1}{K+1}\,\|\yvec\|_2 ,
\label{eq:bound}
\end{equation}
and the bound is approached as $\lambda_{\min} \to 0$.
\end{proposition}

\begin{proof}
At $a_k = (-1)^k$, $s=1$, \eqref{eq:tse} telescopes to
$t_K = \sum_{k=0}^{K} u^k$ with $u = 1-\lambda$, giving \eqref{eq:gain} on
inversion. That sum increases in $u$ and $u$ decreases in $\lambda$, so $g_K$
increases; $g_K(1)=1$; and $0 < 1-u^{K+1} < 1$ for $\lambda<1$ gives
$g_K(\lambda) > \lambda$. Using $1-u^{K+1} = \lambda\sum_{k\le K}u^k$,
\begin{equation}
g_K(\lambda) - \lambda = \frac{\lambda u^{K+1}}{1-u^{K+1}}
= \frac{u^{K+1}}{\sum_{k=0}^{K} u^{k}},
\label{eq:dev}
\end{equation}
which increases in $u$ on $[0,1)$ with limit $1/(K+1)$, so it is bounded by
$1/(K+1)$ throughout. Symmetry of $\Psimat$ turns this into \eqref{eq:bound}.
\end{proof}

\begin{corollary}[Truncation order]
\label{cor:order}
The relative error of $t_K$ as an approximation of $1/\lambda$ is
$(1-\lambda)^{K+1}$, so accuracy $\varepsilon$ uniformly over
$\mathrm{spec}(\Psimat)$ requires
\begin{equation}
K + 1 \;\ge\; \frac{\log(1/\varepsilon)}{-\log(1-\lambda_{\min})}
\;\gtrsim\; \kappa(\Psimat)\log(1/\varepsilon).
\label{eq:Kreq}
\end{equation}
\end{corollary}

\begin{figure}[!t]
\centering
\includegraphics[width=0.65\textwidth]{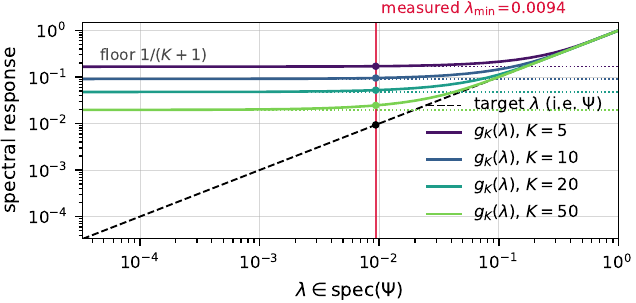}
\caption{Spectral response at initialization. The identity $\lambda$ (the
intended response $\Psimat$), the realized response $g_K(\lambda)$ of
\eqref{eq:gain} for several $K$, and the floor $1/(K+1)$. The measured
$\lambda_{\min}$ of Section~\ref{sec:numerics} is marked.}
\label{fig:gain}
\end{figure}

\begin{figure}[!t]
\centering
\includegraphics[width=0.72\textwidth]{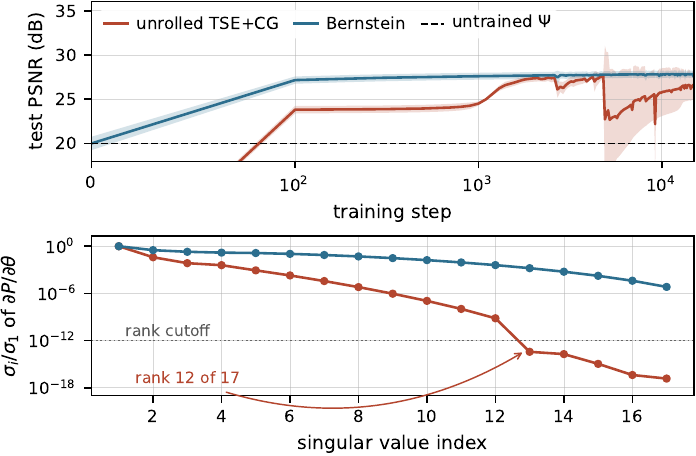}
\caption{Training dynamics at an equal budget of $17$ parameters, three
seeds, each arm at its own best learning rate. \emph{Top:} test PSNR against
step (symlog); dashed line is the untrained $\Psimat$. The arms converge to
nearly the same value, as Theorem~\ref{thm:collapse} implies they must, but
an order of magnitude apart in steps, and the unrolled arm shows abrupt
late-training excursions. That arm enters from below the plotted range: at
initialization its $m=3$ CG module is far from an exact solve ($2.4$~dB at
step $0$), so the response there is not $g_K$.
\emph{Bottom:} singular values of the parameter-to-response Jacobian at
initialization, relative to the largest; the unrolled parameterization is
numerically rank $12$ of $17$, the Bernstein one full rank.}
\label{fig:train}
\end{figure}

\section{A Bernstein Parameterization}
\label{sec:bernstein}

\begin{proposition}
\label{prop:bernstein}
Let $P(\lambda) = \sum_{j=0}^{d} c_j \binom{d}{j}\lambda^j(1-\lambda)^{d-j}$.
Then: (i) as $\{c_j\}$ ranges over $\R^{d+1}$, $P$ spans all of $\Poly_d$;
(ii) if $c_j \in [0,1]$ for all $j$ then $P(\Psimat)$ is positive semidefinite
and $\|P(\Psimat)\|_2 \le 1$, for every $\Psimat$ with
$\mathrm{spec}(\Psimat) \subset [0,1]$; and (iii) $c_j = j/d$ gives
$P(\lambda) = \lambda$ exactly, i.e. the target $\Psimat\yvec$ lies in the
hypothesis class at zero approximation error.
\end{proposition}

\begin{proof}
(i) Conversion to the monomial basis is triangular with nonzero diagonal,
$b_{j,d}$ having lowest term $\binom{d}{j}\lambda^j$~\cite{lorentz1953}.
(ii) On $[0,1]$ the $b_{j,d}$ are nonnegative and
$\sum_j b_{j,d}(\lambda) = (\lambda+(1-\lambda))^d = 1$, so $P(\lambda)$ is a
convex combination of the
$c_j$ and lies in $[0,1]$; as $\Psimat$ is symmetric,
$\mathrm{spec}(P(\Psimat)) = P(\mathrm{spec}(\Psimat)) \subset [0,1]$, which
is PSD together with $\|P(\Psimat)\|_2 \le 1$. (iii) For
$X \sim \mathrm{Binomial}(d,\lambda)$, $\Pr[X=j] = b_{j,d}(\lambda)$, so
$\sum_j (j/d)\,b_{j,d}(\lambda) = \mathbb{E}[X]/d = \lambda$.
\end{proof}

Applications of the black-box $\Psimat$ dominate the cost. Any $P \in \Poly_d$ costs $d$ applications through the Krylov basis
$\{\Psimat^k\yvec\}_{k \le d}$, whatever basis expresses its coefficients,
whereas an unrolled block spends $\Theta(Km)$ to reach degree at most
$K(m-1)$: matching degrees costs no more, and the saving is not the point. By Proposition~\ref{prop:measure} the unrolled block reaches a
measure-zero subset of $\Poly_{K(m-1)}$, closed under scaling and under
addition of constants but not under addition, while the Bernstein form has
all of $\Poly_d$. The comparison is thus not ``the same class more cheaply'' but ``a
strictly larger reachable set at no greater cost, with the guarantee of
Proposition~\ref{prop:bernstein}(ii) on demand''. That guarantee is not free:
$c \in [0,1]^{d+1}$ is sufficient for $P([0,1]) \subseteq [0,1]$ but not
necessary, so it forgoes part of $\Poly_d$; Section~\ref{sec:numerics}
measures the cost.

\begin{remark}[Both are exact solves, of different problems]
\label{rem:map}
By the correspondence of~\cite{chan2019}, any symmetric, positive definite,
nonexpansive $P(\Psimat)$ --- in particular every member of the class of
Proposition~\ref{prop:bernstein}(ii) with positive coefficients --- is the
exact MAP solution for the prior $\Lmat' = \mu^{-1}(P(\Psimat)^{-1}-\Imat)$
--- the same reading of a linear denoiser as a graph filter that
\cite{gdd2024,thm1source2023} quote as Theorem~1. A
truncated-CG block approximately solves the stated problem; the direct filter
exactly solves a sibling problem: neither carries more MAP semantics.
\end{remark}

Linearity also fixes the objective: the output is linear in $c$, so a
squared-error loss is convex in $c$ and stays convex under the box of
Proposition~\ref{prop:bernstein}(ii), its minimizer a (constrained)
least-squares solution. The unrolled objective is not convex, its parameters
entering through a composition. This gives the ablation an optimizer-free
reference.

\begin{remark}[Trainability]
\label{rem:train}
The parameterizations also differ in conditioning. The derivative of the
response with respect to an inner coefficient $a_k$ carries a factor
$q_{m-1}'(t_K(\lambda))$ evaluated where $t_K(\lambda) \approx 1/\lambda$, so
it spans a range widening with $\kappa(\Psimat)$, and rescalings of $t_K$
absorbed by $q_{m-1}$ leave the response unchanged. This is reported as an observation, measured in
Section~\ref{sec:numerics}, not a claim about what optimization does.
\end{remark}

\section{Numerical Results}
\label{sec:numerics}

Non-overlapping $64\times64$ grayscale patches from the 300 images of
TAMPERE17~\cite{ponomarenko2018}; $\Psimat$ from bilateral-style weights on
$(x, y, I, \partial_x I, \partial_y I)$ over a window of radius $r = 3$,
symmetrically normalized, $\mu = 100$; medians over 256 patches at a fixed
seed. \emph{Conditioning:} $\lambda_{\min} = 9.37\times10^{-3}$ and
$\kappa(\Psimat) = 107$ (spread $[104,107]$), so by \eqref{eq:Kreq} accuracy
$\varepsilon = 10^{-2}$ needs $K \ge 488$. By \eqref{eq:kappa} that size rests
on the sign of $w_{\min}$, which the graph construction alone decides.

\begin{proposition}
\label{prop:wmin}
Let the spatial weights be $e^{-a(p-q)^2}$ truncated to $|p-q| \le r$, and set
\begin{equation}
\sigma_r(\theta) = 1 + 2\textstyle\sum_{k=1}^{r} e^{-ak^2}\cos(k\theta).
\label{eq:symbol}
\end{equation}
If $\sigma_r(\pi) < 0$ then, for all $n$ above an explicit $O(1)$ threshold,
$w_{\min} < 0$ and hence $\kappa(\Psimat) > 1 + \mu$.
\end{proposition}

\begin{proof}
$\mathbf{W} = \mathbf{S}^{-1/2}\mathbf{B}\mathbf{S}^{-1/2}$ is a congruence
of $\mathbf{B}$, so by Sylvester's law of inertia both have equally many
negative eigenvalues and it suffices to exhibit $x$ with
$x^{\!\top}\mathbf{B}x < 0$. Absent photometric variation the features reduce
to $(x,y)$, so weights and window factor:
$\mathbf{B} = \mathbf{T}\otimes\mathbf{T}$ and
$\mathbf{W} = \mathbf{V}\otimes\mathbf{V}$ with
$\mathbf{V} = \mathbf{D}^{-1/2}\mathbf{T}\mathbf{D}^{-1/2}$,
$\mathbf{D} = \mathrm{diag}(\mathbf{T}\mathbf{1})$; as
$\lambda_{\max}(\mathbf{V}) = 1$ by Lemma~\ref{lem:kappa},
$w_{\min} = \lambda_{\min}(\mathbf{V})$. Take $x_p = (-1)^p$; since
$\#\{(p,q): p-q = k\} = n - |k|$ and $x_p^2 = 1$,
\begin{equation}
\frac{x^{\!\top}\mathbf{T}x}{x^{\!\top}\mathbf{D}x}
= \frac{n\,\sigma_r(\pi) - \sum_{|k|\le r}|k|\,e^{-ak^2}(-1)^k}
       {\mathbf{1}^{\!\top}\mathbf{T}\mathbf{1}},
\label{eq:rayleigh}
\end{equation}
whose numerator turns negative once $n$ passes the threshold; it bounds
$\lambda_{\min}(\mathbf{V})$ above, and \eqref{eq:kappa} gives the claim.
\end{proof}

Three remarks fix the scope. First, an \emph{untruncated} Gaussian kernel
matrix is positive semidefinite, so a window is necessary for
$w_{\min} < 0$: by Poisson summation
$\sigma_\infty(\pi) = \sqrt{\pi/a}\,2e^{-\pi^2/4a}$, which at $a = 1/9$ is
$2.4\times10^{-9}$ --- positive, but so narrowly that dropping the tail beyond
$r = 3$ carries it to $\sigma_3(\pi) = -0.243$. It is the window, not the
weights. Second, $\sigma_r(\pi) < 0$ is sufficient but not necessary,
$\theta = \pi$ being one test direction; the sharp criterion
$\min_\theta \sigma_r(\theta) < 0$ matched $\mathrm{sign}(w_{\min})$ in all
$42$ configurations checked over $a^{-1/2} \in [0.3,10]$, $r \in [1,12]$, and
for bandwidths narrow relative to the window the graph is near-diagonal with
$w_{\min} > 0$. Third, the windowing is not a defect one could engineer away:
compactly supported positive-definite kernels~\cite{wendland1995} give
$w_{\min} \ge 0$ by construction, but by \eqref{eq:kappa} that lowers $\kappa$
only from $1+1.047\mu$ to $1+\mu$. The required order is the price of $\mu$,
not of kernel engineering.

Here \eqref{eq:rayleigh} gives $w_{\min} \le -0.0469$, so
$\kappa \ge 1 + 1.047\mu$ against the measured $107$ at $\mu = 100$.
Consistently, a \emph{constant} patch
reproduces the median $\kappa$ to within $0.34\%$, and a factor-$25$ change of
intensity bandwidth moves it by under $0.5\%$: the required order follows from
$\mu$ and the graph, not the dataset. Across two decades of $\mu$ --- $10$, $100$, $1000$ --- \eqref{eq:kappa} gives
$\kappa = 11.6$, $107$, $1058$ and, by \eqref{eq:Kreq}, required orders
$K \ge 50$, $488$, $4869$: the conclusion does not rest on the choice of
$\mu$.

\emph{Spectral response.} Fig.~\ref{fig:gain}. At the measured
$\lambda_{\min}$ and $K = 10$, the realized gain is
$g_{10}(\lambda_{\min}) = 9.52\times10^{-2}$ against a target of
$9.37\times10^{-3}$ --- a factor of $10.2$, within $5\%$ of the floor $1/(K{+}1)$.

\emph{Ablation.} Three arms denoise patches at $\sigma = 0.1$ on the unit
intensity range, equivalently $\sigma \approx 25.5$ at eight bits:
(a) untrained $\Psimat$; (b) the unrolled TSE+CG block at $K = 10$, $m = 3$,
learning $\{a_k\}, s, \{\alpha_k\}, \{\beta_k\}$ --- $17$ parameters; and
(c) a Bernstein polynomial of degree $16$, learning its $17$ coefficients from
the initialization $c_j = j/d$ of Proposition~\ref{prop:bernstein}(iii). Both
learned arms are polynomials in $\Psimat$ by Theorem~\ref{thm:collapse}, so
this compares two parameterizations of one class, not two architectures; each
gets its own best rate from a common sweep ($5\times10^{-2}$,
$2\times10^{-1}$, selected on training loss) and is trained to a plateau
($15\,000$ steps, $24$ training and $24$ held-out patches). Over
three seeds, each redrawing patches and noise,
\par\smallskip
\centerline{%
\begin{tabular}{lcc}
\hline
arm & parameters & test PSNR (dB) \\
\hline
(a) untrained $\Psimat$ & 0  & $20.00 \pm 0.76$ \\
(b) unrolled TSE+CG     & 17 & $26.65 \pm 1.67$ \\
(c) Bernstein           & 17 & $27.52 \pm 0.60$ \\
\hline
best $P \in \Poly_{20}$ & 21 & $28.08 \pm 0.42$ \\
\hline
\end{tabular}}
\par\smallskip\noindent
The last row is a reference, not an arm: the training-objective least-squares
optimum over the ambient class, scored on the test patches. Arm (c) lands
$0.6$~dB below it and arm (b) $1.4$~dB, their means
differing by less than the spread of arm (b) --- what
Theorem~\ref{thm:collapse} leads one to expect, and why the comparison below
is about optimization, not accuracy. The guarantee of
Proposition~\ref{prop:bernstein}(ii) is not free either: the same convex
problem over $c \in [0,1]^{17}$ gives $25.64 \pm 0.53$~dB against
$27.99 \pm 0.42$ unconstrained in the same class --- about $2.4$~dB, the box
being sufficient for $P([0,1]) \subseteq [0,1]$ but not necessary.

\emph{Training dynamics.} What separates the arms is how they get there, not
where they arrive. Arm (c) reaches $25.65$~dB within $100$ steps for every
seed, against $1200$--$1300$ for arm (b) (Fig.~\ref{fig:train}, top). Arm (b)
is also less stable, across seeds ($\pm 1.67$ against $\pm 0.60$~dB) and
across the learning rate: over the grid
$\{5\!\times\!10^{-3},\dots,2\!\times\!10^{-1}\}$ its final PSNR moves by
$3.95$~dB against $0.29$~dB for arm (c).

Gradient norms are not
comparable across arms, each carrying the units of its own parameters, but the
Jacobian of $\theta \mapsto P(\cdot)|_{\mathrm{spec}(\Psimat)}$ is. At
initialization (Fig.~\ref{fig:train}, bottom) its singular values for arm (c)
span $10^{5}$ at full rank $17$, whereas for arm (b) they fall below relative
$10^{-12}$ past the twelfth: seventeen parameters purchase twelve usable
directions, the twelve $(\{a_k\},s)$ carrying the eleven degrees of freedom of
$t_K$ and the five $(\{\alpha_k\},\{\beta_k\})$ the three of $q_{m-1}$;
invariances of the composition, such as a rescaling of $t_K$ absorbed by
$q_{m-1}$, reduce those fourteen to the observed twelve. Rank
$12$ of $17$ sits strictly inside the dimension bound \eqref{eq:thin} of
Proposition~\ref{prop:measure}, and is Remark~\ref{rem:train} made
quantitative.

\section{Conclusion}

Unrolling a truncated expansion inside an unrolled solver produces a
polynomial graph filter of degree at most $K(m-1)$ at every point of training,
with a reachable set measure-zero within even that class. The destination is
fixed by the operator, as Theorem~\ref{thm:collapse} leads one to expect; the
parameterizations differ in the path --- conditioning, not capacity: an order of
magnitude in steps, a rank-$12$-of-$17$ Jacobian, near-equal accuracy.
Architectures re-estimating the graph per block compose filters in
\emph{varying} operators --- an open class.

\bibliographystyle{IEEEtran}
\bibliography{refs}

\end{document}